\let\accentvec\vec % necessary to inhibit amsmath warning caused by definition of \vec in llncs
\documentclass[envcountsame]{llncs}

\let\vec\accentvec

\usepackage{\jobname}

\spnewtheorem*{claimproof}{Proof of Claim}{\bfseries}{\normalfont}

\title{%
  \bfseries\Large%
  On the Parameterized Complexity of\\
Default Logic and Autoepistemic Logic%
 \thanks{Work supported by DFG grant VO 630/6-2.}
}

\author{
  Arne Meier\inst{1}\and
  Johannes Schmidt\inst{2}\and
  Michael Thomas\inst{3,}\thanks{Work performed while employed at the Gottfried Wilhelm Leibniz Universit\"{a}t, Appelstr.~4,
	30167~Hannover, Germany}\and
  Heribert Vollmer\inst{1}
}
\institute{
Universit\"at Hannover~ \email{$\{$meier,vollmer$\}$@thi.uni-hannover.de}
\and Université de la Méditerranée~ \email{johannes.schmidt@lif.univ-mrs.fr}
\and TWT GmbH~ \email{michael.thomas@twt-gmbh.de}
}

\newif\ifreport

\reporttrue

\begin{document}

\maketitle

\begin{abstract}
We investigate the application of Courcelle's Theorem and the logspace version of Elberfeld \etal{} in the context of the implication problem for propositional sets of formulae, the extension existence problem for default logic, as well as the expansion existence problem for autoepistemic logic and obtain fixed-parameter time and space efficient algorithms for these problems.

On the other hand, we exhibit, for each of the above problems, families of instances of a very simple structure that, for a wide range of different parameterizations, do not have efficient fixed-parameter algorithms (even in the sense of the large class $\XPnu$), unless $\P=\NP$.

%Further we introduce the notion of pseudo-cliques as a method for bounding the tree width of a graph from below. 

%Furthermore, if theses decision problems for several kinds of restricted default theories, sets of autoepistemic formulae, and sets of propositional formulae can be efficiently answered (\ie, its slices are in $\P$) under a constant parameterization, we can prove several complexity class collapses.
\end{abstract}

\section{Introduction} \label{sect:intro}

Non-monotonic reasoning formalisms were introduced in the 1970s as a formal model for human reasoning and have developed into one of the most important topics in computational logic and artificial intelligence. However, as it turns out, most interesting reasoning tasks are computationally intractable already for propositional versions of non-monotonic logics \cite{gottlob92}, in fact presumably much harder than for classical propositional logic. Because of this, a lot of effort has been spent to identify fragments of the logical language for which at least some of the algorithmic problems allow efficient algorithms; a survey of this line of research can be found in \cite{thvo10}.

In this paper a different approach is chosen to deal with hard problems, namely the framework of parameterized complexity. Gottlob et al.{} \cite{gopiwe10} made it clear that the tree width of a (suitable graph theoretic encoding of a) given knowledge base is a useful parameter in this context: making use of Courcelle's Theorem it was shown that many reasoning tasks for logical formalisms such as circumscription, abduction and logic programming become tractable in the case of bounded tree width. Here we examine a family of non-monotonic logics where the semantics of a given set of formulae (axioms, knowledge base) is defined in terms of a fixed-point equation. In particular we turn to default logic \cite{reiter:80} and autoepistemic logic \cite{moore:85}. In the first, human reasoning is mimicked using so called ``default rules'' (in the absence of contrary information, assume this and that); in the second, a modal operator is introduced to model the beliefs of a perfect rational agent. For both logics the algorithmic tasks of satisfiability and reasoning have been shown to be complete in the second level of the polynomial hierarchy \cite{gottlob92}. 

Much in the vein of \cite{gopiwe10} we here examine the parameterized complexity of these problems and, making again use of Courcelle's Theorem and a recent improvement by Elberfeld et al., we obtain time and space efficient algorithms if the tree width of the given knowledge base is bounded. 
This proves once again how important this parameter is. 

A second contribution of our paper concerns lower bounds: Under the assumption $\P\neq\NP$ we show that, even for certain families of very simple knowledge bases and for any parameterization taken from a broad variety, no efficient fixed-parameter algorithms exist, not even in the sense of the quite large parameterized class $\XPnu$. These simple families of knowledge bases are defined in terms of severe syntactic restrictions, e.g., using default rules with literals or propositions only. Restricting the input structure even further we obtain that no fixed-parameter algorithm in the sense of the space-bounded class $\XLnu$ (the logarithmic space analogue of $\XPnu$) exists, unless $\LOGSPACE=\NL$.

Unfortunately, tree width is not among the parameters for which our lower bound can be proven---otherwise we would have proven $\P\neq\NP$. In a third part of our paper, we show that those structurally very simple families of knowledge bases, for which we gave our lower bounds, already have unbounded tree width. For this result, we introduce the notion of pseudo-cliques and show how to embed these into our graph-theoretic encodings of knowledge bases.

\ifreport
\else
Due to space reasons the proof of Theorem~\ref{thm:tw-pseudo-cliques-cardinality} has to be omitted and can be read in the CoRR version "Meier, A., Schmidt, J., Thomas, M., Vollmer, H.: \emph{On the parameterized complexity
of default logic and autoepistemic logic.} CoRR abs/1110.0623 (2011)".
\fi

\section{Preliminaries} \label{sect:prelim}

\paragraph{Complexity Theory.}

In this paper we will make use of several standard notions of complexity theory such as the complexity classes $\LOGSPACE$, $\ParityL$, $\NL$, $\P$, $\NP$, $\coNP$, and $\sigPtwo$
and their completeness notions under logspace-many-one $\leqlogm$ reductions.

%\paragraph{Parameterized Complexity.} 

Given a problem $P$ and a parameterization $\kappa$, $(P,\kappa)$ belongs to the class $\FPT$ iff there is a deterministic algorithm solving $P$ in time $f(\kappa(x))\cdot|x|^{O(1)}$; $(P,\kappa)$ is said to be \emph{fixed parameter tractable} then. If $(P,\kappa)$ is a parameterized problem, then $(P,\kappa)_\ell\eqdef\{x\in P\mid \kappa(x)=\ell\}$ is the $\ell$-th slice of $(P,\kappa)$. Define $(P,\kappa)$ to be a member of $\XPnu$ (in words, $\XP$ nonuniform) iff $(P,\kappa)_\ell\in\P$ for all $\ell\in\N$. For background on parameterized complexity we recommend \cite{flgr06}.

Furthermore, we require space parameterized complexity classes which have been defined in \cite{Stockhusen11} recently. Given a parameterized problem $(P,\kappa)$, we say
\begin{itemize}
	\item $(P,\kappa)\in\PLS$ iff there exists a deterministic algorithm deciding $P$ in space $O(\log{\kappa(x)}+\log|x|)$,
	\item $(P,\kappa)\in\FPL$ iff there exists a deterministic algorithm deciding $P$ in space $O(\log{f(\kappa(x))}+\log|x|)$ for some recursive funktion $f$, and
	\item $(P,\kappa)\in\XLnu$ iff $(P,\kappa)_\ell\in\LOGSPACE$ for all $\ell\in\N$.
\end{itemize}
It holds that $\PLS\subseteq\FPL\subseteq\FPT\subseteq\XPnu$ as well as $\FPL\subseteq\XLnu\subseteq\XPnu$.

\paragraph{Tree width.}

A \emph{tree decomposition} of a graph $G = (V,E)$ is a pair $(T,X)$, where $X = \{B_1,\dots,B_r\}$ is a family of subsets of $V$
(the set of \emph{bags}), and $T$ is a tree whose nodes are the bags $B_i$, satisfying the following conditions:
\emph{(i)} $\bigcup X = V$, i.e., every node appears in at least one bag, \emph{(ii)} $\forall (u,v) \in E$ $\exists B\in X$: $u,v \in B$, i.e., every edge is 'contained' in a bag, and \emph{(iii)} $\forall u\in V$: $\{B\; |\; u \in B\}$ is connected in $T$, i.e., for every node $u$ the set of bags containing $u$ is connected in $T$.

The \emph{width} of a decomposition $(T,X)$, $\width(T,X)$, is the number $\max\{\,|B| \;\big|\; B \in X\} - 1$, \ie, the size of the largest bag minus 1.
The \emph{tree width} of a graph $G$, $\tw(G)$, is the minimum of the widths of the tree decompositions of $G$.

% Clones werden nicht ben\"otigt.
%
%\paragraph{Clones.} 
%
%Let $B$ be a finite set of Boolean functions. $B$ is called a \textit{clone} if it is closed under superposition (\ie, arbitrary compositions and projections of functions). With $[B]$ we denote the smallest clone containing $B$. If $B$ is a set of Boolean functions, then a $B$-formula $\varphi$ only uses functions from $B$.

\paragraph{Default Logic.} 

Following Reiter \cite{reiter:80}, a \emph{default rule} is a triple $\frac{\alpha:\beta}{\gamma}$; $\alpha$ is called the \textit{prerequisite}, $\beta$ is called the \textit{justification}, and $\gamma$ is called the \textit{conclusion}. If $B$ is a set of Boolean functions, then $d=\frac{\alpha:\beta}{\gamma}$ is a $B$-default rule if $\alpha,\beta,\gamma$ are $B$-formulae, i.e., formulae that use only connectors for functions in $B$.
A $B$-\textit{default theory} $(W,D)$ consists of a set of propositional $B$-formulae $W$ and a set of $B$-default rules $D$.

Let $(W,D)$ be a default theory and $E$ be a set of formulae. Now define $\Gamma(E)$ as the smallest set of formulae such that \emph{(i)} $W\subseteq\Gamma(E)$, \emph{(ii)} $\Gamma(E)$ is closed under deduction, \ie, $\Gamma(E)=\Th{\Gamma(E)}$, and \emph{(iii)} for all defaults $\frac{\alpha:\beta}{\gamma}\in D$ with $\alpha\in\Gamma(E)$ and $\lnot\beta\notin E$, it holds that $\gamma\in\Gamma(E)$. Then a \emph{stable extension} of $(W,D)$ is a fix-point of $\Gamma$, \ie, a set $E$ such that $E=\Gamma(E)$.

A definition for stable extensions beyond fix-point semantics which was introduced by Reiter \cite{reiter:80} uses the principle of a stage construction: for a given default theory $(W,D)$ and a set $E$ of formulae, define $E_0=W$ and $E_{i+1} = \Th{E_i}\cup\{\gamma\mid \frac{\alpha:\beta}{\gamma}\in D,\alpha\in E_i\text{ and } \lnot\beta\notin E\}$. Then $E$ is a \emph{stable extension} of $(W,D)$ if and only if $E=\bigcup_{i\in\N}E_i$, and the set $G=\{\frac{\alpha:\beta}{\gamma}\in D\mid \alpha\in E\land\lnot\beta\notin E\}$ is called the set of \emph{generating defaults}.

The so to speak satisfiability problem for default logic, here called \emph{extension existence problem}, $\EXT$, is the problem, given a theory $(W,D)$, to decide if it has a stable extension. Gottlob \cite{gottlob92} proved that this problem is complete for the class $\sigPtwo$.

\paragraph{Autoepistemic Logic.} 

Moore in 1985 introduced a new modal operator $L$ stating that its argument is "believed" as an extension of propositional logic \cite{moore:85}. Further the expression $L\varphi$ is treated as an atomic formula with respect to the consequence relation $\models$. Given a set of Boolean functions $B$, we define with $\LAE(B)$ the set of all autoepistemic $B$-formulae through $\varphi::= p\mid f(\varphi,\dots,\varphi)\mid L\varphi$ for $f$ being a Boolean functions in $B$ and a proposition $p$. If $\Sigma\subseteq\LAE(B)$, then a set $\Delta\subseteq\LAE(B)$ is a \emph{stable expansion} of $\Sigma$ if it satisfies the condition $\Delta=\Th{\Sigma\cup L(\Delta)\cup\lnot L(\overline\Delta)}$, where $L(\Delta)\eqdef\{L\varphi\mid\varphi\in\Delta\}$ and $\lnot L(\overline\Delta)\eqdef\{\lnot L\varphi\mid\varphi\notin\Delta\}$, and $L(\Delta),\lnot L(\overline\Delta)\subseteq\LAE(B)$. 

Let $\SF\varphi$ ($\SFL\varphi$) define the set of ($L$ preceded) subformulae for an autoepistemic formula $\varphi$, and let us use the shorthand $\lnot S=\{\lnot\varphi\mid\varphi\in S\}$ for a set of (autoepistemic) formulae $S$. Given a set of autoepistemic $B$-formulae $\Sigma\subseteq\LAE(B)$, we say a set $\Lambda\subseteq\SFL\Sigma\cup\lnot\SFL\Sigma$ is \emph{$\Sigma$-full} if for each $L\varphi\in\SFL\Sigma$ it holds $\Sigma\cup\Lambda\models\varphi$ iff $L\varphi\in\Lambda$, and $\Sigma\cup\Lambda\not\models\varphi$ iff $\lnot L\varphi\in\Lambda$. 

The connection of $\Sigma$-full sets and stable expansions of $\Sigma$ has been observed by Niemelä \cite{niemela:90}: if $\Sigma\subseteq\LAE$ is a set of autoepistemic formulae and $\Lambda$ is a $\Sigma$-full set, then for every $L\varphi\in\SFL\Sigma$ either $L\varphi\in\Lambda$ or $\lnot L\varphi\in\Lambda$. The stable expansions of $\Sigma$ and $\Sigma$-full sets are in one-to-one correspondence.

The \emph{expansion existence problem}, $\EXP$, is the problem, given a set of autoepistemic formulae $\Sigma$, to decide if it has a stable expansion. Again, Gottlob proved that this problem is complete for the class $\sigPtwo$.

\section{MSO-Encodings} \label{sect:mso-enc}

It will be the aim of this paper to apply Courcelle's theorem 
%and a recent improvement 
to obtain fixed-parameter algorithms in the context of default and autoepistemic logic. For this, we will have to describe the relevant decision problems by monadic second-order formulae. In this section we will explain how to do this and obtain a preliminary result for the implication problem. Our approach is similar to the one of Gottlob, Pichler, and Wei \cite{gopiwe10} where MSO encodings for algorithmic problems from logic programming, abduction, and circumscription where developed.

Now fix a finite set $B$ of Boolean functions. 
Denote by $\tau_B$ the vocabulary $\{\const^1_f \mid f \in B, \arity(f) = 0\} \cup \{\conn_{f,i}^2 \mid f \in B, 1 \leq i \leq \arity(f)\}$.
With respect to a set $\Gamma$ of propositional $B$-formulae we associate a $\tau_{B,\textit{prop}}$-structure $\struc{A}_\Gamma$ where $\tau_{B,\textit{prop}}\eqdef \tau_B\cup\{\var^1,\repr^1\}$ such that the universe of 
$\struc{A}_\Gamma$ is the set of subformulae of the formulae in $\Gamma$, and \emph{(i)} $\var(x)$ holds iff $x$ represents a variable, \emph{(ii)} $\repr(x)$ holds iff $x$ represents a formula from $\Gamma$, \emph{(iii)} $\const_f^1(x)$ holds iff $x$ represents the constant $f$, and \emph{(iv)} $\conn^2_{f,i}(x,y)$ holds iff $x$ represents the $i$th argument of the function $f$ at the root of the formula tree represented by $y$.
%The construction of $\struc{A}_\Gamma$ clearly applies for a single $B$-formula $\varphi$, too.

\begin{lemma} \label{lem:mso-sat}
  Let $B$ be a finite set of Boolean functions.
  Then there exists an MSO-formula $\theta_{\textit{sat}}$ over $\tau_{B,\textit{prop}}$ such that for any $\Gamma \subseteq \allFormulae(B)$ it holds that
  $\Gamma \text{ is satisfiable}$ $\text{iff } \struc{A}_\Gamma \models \theta_{\textit{sat}}.$
\end{lemma}
\begin{proof}
	We define the MSO-formulae $\theta_{\textit{struc}}$, $\theta_{\textit{assign}}$, and $\theta_{\exists\textit{assign}}$ over the vocabulary $\tau_{B,\textit{prop}}$ as follows:
  \begin{multline*}
    \theta_{\textit{struc}} \eqdef \,
          % eindeutige Vorgänger
           \forall x \Big(\lnot\repr(x)\to \exists y \big(\lnot \var(y)\land\bigvee_{\mathclap{\substack{f \in B,\\ 1 \leq i \leq \arity(f)}}} \conn_{f,i}(x,y)\big)\Big)\land
          % eindeutige Nachfolger
          \forall x \Big(\lnot \var(x) \imp \\
            \bigvee_{\mathclap{\substack{f \in B,\\ \arity(f)=0}}} \const_f(x) \xor
            \bigvee_{f \in B}\hspace{-1em}\bigwedge_{\substack{\phantom{\bigvee}\\1 \leq i \leq \arity(f)}}\hspace{-1.7em}\exists y \big(\conn_{f,i}(y,x)\land\forall z(\conn_{f,i}(z,x)\to z=y)\big)\Big).
  \end{multline*}
  
  The formula $\theta_{\textit{struc}}$ states that if an individual is not representing a formula $\varphi\in\Gamma$, then there must be at least one subformula in which it occurs.
 If an individual is not a variable, then it represents either a constant or a Boolean function $f\in B$ and needs to have corresponding $\arity f$ individuals.

  Let $n$ denote the maximal arity of $B$, i.e., $n\eqdef \max\{\arity(f)\mid f\in B\}$.
  \begin{multline*}
    \theta_{\textit{assign}}(M) \eqdef \, 
           \forall x \forall y_1 \cdots \forall y_n \bigwedge_{f \in B} 
            \Big(
              ~~~\bigwedge_{\mathclap{\substack{f\in B,\\\arity(f)=0}}}\const_f(x)\to (M(x) \leftrightarrow f)\land\\
              \bigwedge_{\mathclap{1 \leq i \leq \arity(f)}}~~~ \conn_{f,i}(y_i,x) \imp 
              \big(
                M(x) \leftrightarrow f(\llbracket y_1 \in M\rrbracket,\ldots,\llbracket y_{\arity(f)} \in M\rrbracket)
              \big)
            \Big),
	\end{multline*}
	where $\llbracket x\in M\rrbracket$ is $\true$ iff $x\in M$ holds and $\false$ otherwise. Now define
	\begin{align*}            
    \theta_{\exists\textit{assign}} &\eqdef \, 
          \exists M \big( \theta_{\textit{assign}}(M) \land \forall x\Big(\repr(x) \imp M(x) \Big) 
  \end{align*}
  It is easy to verify that $\theta_{\textit{sat}} \eqdef  \theta_{\textit{struc}} \land \theta_{\exists\textit{assign}}$ satisfies the lemma.\qed
\end{proof}

\noindent Let $B$ be a finite set of Boolean functions and $F,G$ be sets of $B$-formulae. Answering the implication problem of sets of propositional formulae, i.e., the question whether $F\models G$, requires to extend our vocabulary $\tau_{B,\textit{prop}}$ to $\tau_{B,\textit{imp}}\eqdef \tau_{B,prop}\cup\{\reprPrem^1, \reprConc^1\}$ as well as our structure which we will denote %$\tau_{B,\textit{imp}}$-structure
by $\struc{A}_{F,G}$: $\reprPrem(x)$ is true iff $x$ represents a formula from $F$, and $\reprConc(x)$ is true iff $x$ represents a formula from $G$. Now it is straightforward to formalize implication.
%: for any unary predicate $c^1$, extend $\struc{A}_\Gamma$,
%as constructed above, to hold for exactly those elements representing the formulae from 
%$\Gamma$ that are to be tested for implication. Denote the resulting structure by $\struc{A}_{\Gamma,r}$

\begin{lemma}\label{lem:mso-imp}
  Let $B$ be a finite set of Boolean functions.
  Then there exists an MSO-formula $\theta_{\textit{imp}}$ over $\tau_{B,\textit{imp}}$ such that for any $\Gamma \subseteq \allFormulae(B)$ and any
%  $\Delta_1,\Delta_2\subseteq\Gamma$ it holds that 
  $F,G\subseteq\Gamma$ it holds that 
  $F \models G \text{ iff } \struc{A}_{F,G} \models \theta_{\textit{imp}}.$
\end{lemma}
\begin{proof}
  Define the MSO-formulae $\theta_{\textit{premise}}(M)$, $\theta_{\textit{conclusion}}(M)$, and $\theta_{\textit{implies}}$ as follows:
  \begin{align*}
    \theta_{\textit{premise}}(M) \eqdef  \, & \forall x (\reprPrem(x)\imp M(x)) \\
    \theta_{\textit{conclusion}}(M) \eqdef  \, & \forall x (\reprConc(x) \imp M(x))) \\
    \theta_{\textit{implies}} \eqdef  \, & \forall M \Big(\big(\theta_{\textit{assign}}(M) \land \theta_{\textit{premise}}(M)\big) \imp \theta_{\textit{conclusion}}(M)\Big)
  \end{align*}
  Then, we can define the formula $\theta_{\textit{imp}}$ as $\theta_{\textit{imp}}\eqdef \theta_{\textit{struc}} \land \theta_{\textit{implies}}$, where $\theta_{\textit{struc}}$ and $\theta_{\textit{assign}}$ are defined as above in Lemma~\ref{lem:mso-sat}.\qed
\end{proof}

The application of Courcelle's Theorem~\cite{courcelle90} and the logspace version of Elberfeld \etal\,\cite{ebjata10} directly leads to the following theorem.

\begin{theorem} \label{thm:imp-FPT-for-bounded-tw}
  Let $B$ be a finite set of Boolean functions, 
  let $k \in \N$ be fixed, and let $F,G$ be sets of $B$-formulae such that $\struc{A}_{F,G}$ has tree width bounded by $k$. 
  Then the implication problem for sets of $B$-formulae is solvable in time $O(f(k)\cdot(|F|+|G|))$ and space $O(\log(|F|+|G|))$.
\end{theorem}

In other words, the implication problem of sets of formulae parameterized by the tree width of $\struc{A}_{F,G}$ is fixed-parameter tractable, and even in $\PLS$. In the following sections we will extend this result to default logic and autoepistemic logic.

\section{Default Logic}

Let $B$ be a finite set of Boolean functions. 
Write $W \Cup D$ as a shorthand for the set of formulae $W \cup \{\alpha,\beta,\gamma \mid \frac{\alpha:\beta}\gamma \in D\}$.
To any $B$-default theory $(W,D)$, we associate a 
$\tau_{B,\textit{dl}} \eqdef  \tau_{B,\textit{prop}} \cup \{\kb^1,\default^1,\prem^2,\just^2,\concl^2\}$-structure $\struc{A}_{(W,D)}$ 
such that the universe of $\struc{A}_{(W,D)}$ is the union of the set of subformulae of 
$W \Cup D\cup\{\lnot\beta\mid\frac{\alpha:\beta}{\gamma}\in D\}$ together with a set corresponding to the defaults in $D$, the relations from $\tau_{B,\textit{prop}}$ are interpreted as in Section~\ref{sect:mso-enc}, and  
\begin{itemize}
	\item $\kb(x)$ holds iff $x$ represents a formula from the knowledge base $W$,
  \item $\default(x)$ holds iff $x$ represents a default $d \in D$,
  \item $\prem(x,y)$ (resp.\ $\just(x,y)$, $\concl(x,y)$) holds iff 
  $x$ represents the premise $\alpha$ (resp.\ justification $\beta$, conclusion $\gamma$) and $y$ represents the default rule $\frac{\alpha:\beta}\gamma$.
\end{itemize}

\begin{lemma} \label{lem:mso-ext}
  Let $B$ be a finite set of Boolean functions
  and let $(W,D)$ be a $B$-default theory.
  There exists an MSO-formula $\theta_{\textit{extension}}$ such that $(W,D)$ possesses a stable extension 
  iff $\struc{A}_{(W,D)} \models \theta_{\textit{extension}}$.
\end{lemma}
\begin{proof}
	First the formula $\theta_{\textit{isneg}}$ expresses the fact that one formula is the negation of another formula:
	$
		\theta_{\textit{isneg}}(\varphi,\overline\varphi)\eqdef \theta_{\textit{struc}}\land\forall M\Big(\theta_{\textit{assign}}(M)\to\big(M(\varphi)\leftrightarrow \lnot M(\overline\varphi)\big)\Big).
	$
	Observe that $\varphi$ and $\overline\varphi$ are not formulae but placeholders for individuals. 
  The following two formulae define the applicability of defaults, \ie, whether a premise $\alpha$ is entailed or a justification $\beta$ is violated which uses the shortcut $\chi(C,M,x)\eqdef(\kb(x)\lor C(x))\to M(x)$:
  \begin{align*}
  		\theta_{W\cup C\models\alpha}(C,\alpha) &\eqdef \forall M\Big(\theta_{\textit{assign}}(M)\to 
		\forall x\Big(\chi(C,M,x)\to M(\alpha)\Big)\Big),\\
  		\theta_{W\cup C\models\lnot\beta}(C,\beta) &\eqdef\exists\overline\beta\exists M\!\Bigg(\!\theta_{\textit{assign}}(M)\!\to\!
			\forall x\Big(\chi(C,M,x)\land M(\overline\beta)\land \theta_{\textit{isneg}}(\beta,\overline\beta)\!\Big)\!\!\Bigg).
	\end{align*}

 \noindent Now we can define the MSO-formulae $\theta_{\textit{app}}$ (a default $d$ is applicable), $\theta_{\textit{stable}}$ (a set of defaults is stable), $\theta_{\textit{gd}}$ (a set of defaults is generating) as follows.
  \begin{align*}
    \theta_{\textit{app}}(d,G) \eqdef  \, 
      & 
			\exists\alpha\exists\beta\exists C\Big(\prem(\alpha,d)\land \just(\beta,d)\land\\
				\forall x\big(C(x)&\leftrightarrow\exists y(G(y)\land \concl(x,y))\big)\land 
				\theta_{W\cup C\models\alpha}(C,\alpha)\land\lnot\theta_{W\cup C\models\lnot\beta}(C,\beta)\Big)
      \\
    \theta_{\textit{stable}}(G) \eqdef  \, 
      &
      \forall d \big(\default(d)\land(G(d) \leftrightarrow \theta_{\textit{app}}(d,G))\big)
      \\
    \theta_{\textit{gd}}(G) \eqdef  \, 
      & 
      \theta_{\textit{stable}}(G) \land \forall G'(G' \subsetneq G \imp \neg \theta_{\textit{stable}}(G')) 
  \end{align*}
  Then $\theta_{\textit{extension}} \eqdef  \theta_{\textit{struc}} \land \exists G(\theta_{\textit{gd}}(G))$ is true under $\struc{A}_{(W,D)}$ iff $(W,D)$ has a stable extension.\qed
\end{proof}

As a consequence of Lemma~\ref{lem:mso-ext}, we obtain from Courcelle's Theorem~\cite{courcelle90} and the logspace version of Elberfeld \etal\,\cite{ebjata10} that, given the tree width of $\struc{A}_{(W,D)}$ as a parameter, the extension existence problem for default logic is fixed-parameter tractable, and in fact, in $\PLS$.

\begin{theorem} \label{thm:dl-FPT-for-bounded-tw}
  Let $B$ be a finite set of Boolean functions, 
  let $k \in \N$ be fixed, 
  and let $(W,D)$ be a $B$-default theory such that $\struc{A}_{(W,D)}$ has tree width bounded by $k$. 
  Then the extension existence problem for $B$-default logic is solvable in time $O(f(k)\cdot|(W,D)|)$ and space $O(\log|(W,D)|)$.
\end{theorem}

%--hv: Ich verstehe den folgenden Abschnitt nicht, daher unten ein Ersatz:
%
%As a comparison to the restriction imposed by bounding the tree width of the structure $\struc{A}_{(W,D)}$, 
%% suppose that, \eg, the maximal nesting depth of the formulae in $W \cup D$ is bounded.
%% In this case, the satisfiability problem of $E \subseteq W \cup D$ is $\W{1}$-complete. Hence, the extension existence problem is
%% $\W{1}$-hard. 
%% Actually the situation is worse:
%suppose that a given parameterization does not restrict the number of defaults
%(\eg, a parameterization such that satisfiability and implication are fixed-parameter tractable),
%then the extension existence problem is not contained in $\XPnu$ unless $\P = \NP$,
%even if all defaults are composed of literals only.

So again and maybe with no big surprise, similar to the study by Gottlob et al.~\cite{gopiwe10} for different nonmonotonic formalisms, we see here that bounding the tree width of a default theory yields time and space efficient algorithms for satisfiability. In the following we want to contrast this with a strong lower bound. We consider knowledge bases with very simple defaults rules, namely consisting only of literals (and in a second step even only propositions). Then we consider any parameterization of the extension existence problem that is bounded for all knowledge bases that obey this restriction. It follows that even for these very restricted knowledge bases, the parameterized extension existence problem is not even in the class $\XPnu$, unless $\P\neq\NP$.

We want to point out that this theorem comprises for example the usual parameterizations for $\SAT$ (in terms of, e.g., backdoor sets or formula tree width): For all these, we have $\FPT$-algorithms for propositional satisfiability, but still the extension existence problem is not in $\XPnu$.

\begin{theorem} \label{thm:dl-not-in-XP}
  Let $B$ be a finite set of Boolean functions such that $\neg \in [B \cup \{\true\}]$ and
  let $\mathbf{D}$ be the set of sets $D$ of default rules such that each default 
  $d \in D$ is composed of literals only.
  Further let $\kappa$ be a parameterization function
  for which 
  there exists a $c \in \N$ 
  such that $\kappa\big((\emptyset,D)\big) < c$ for all $D \in \mathbf{D}$.
  %Then the extension existence problem for $B$-default logic, parameterized by $\kappa$, is not contained in $\XPnu$ unless $\P = \NP$.
  If $\P\neq\NP$, then the extension existence problem for $B$-default logic, parameterized by $\kappa$, is not contained in $\XPnu$.
\end{theorem}
\begin{proof}
  The reduction from $\SAT$ to default logic restricted to default theories 
  with $W=\emptyset$ and default rules composed of literals only, shown in Lemma 5.6 of \cite{beyersdorffMTV09}, proves that the extension existence problem of default logic restricted to theories of this kind (which will be denoted by $\EXT'$) is $\NP$-hard.
  Now let $\kappa$ be such a parameterization and suppose $\P\neq\NP$. For contradiction assume $(\EXT',\kappa)\in\XPnu$. Hence, by definition of $\XPnu$, it holds $(\EXT',\kappa)_\ell\in\P$ for every $\ell\in\N$. As also $\ell<c$ holds we can compose a deterministic polynomial time algorithm which solves $\EXT'$. This contradicts $\P\neq\NP$ and concludes the proof.\qed
\end{proof}

\begin{theorem} \label{thm:dl-not-in-XL}
  Let $B$ be a finite set of Boolean functions such that $\false \in [B]$ and
  let $\mathbf{D}$ be the set of sets $D$ of default rules such that each default 
  $d \in D$ is composed of propositions or the constant $\false$ only.
  Further let $\kappa$ be a parameterization function
  for which 
  there exists a $c \in \N$ 
  such that $\kappa\big((W,D)\big) < c$ for all $D \in \mathbf{D}$ and all $W$ that consists of at most one proposition.
  %Then the extension existence problem for $B$-default logic, parameterized by $\kappa$, is not contained in $\XPnu$ unless $\P = \NP$.
  If $\LOGSPACE\neq\NL$, then the extension existence problem for $B$-default logic, parameterized by $\kappa$, is not contained in $\XLnu$.
\end{theorem}
\begin{proof}
  The reduction from $\GAP$ to default logic restricted to default theories 
  with $|W|\leq 1$ and default rules composed of propositions or the constant $\false$ only, shown in Lemma 5.8 of \cite{beyersdorffMTV09}, proves that the extension existence problem of default logic restricted to theories of this kind (which will be denoted by $\EXT'$) is $\NL$-hard.
  Following the argumentation in the proof of Theorem~\ref{thm:dl-not-in-XP}, we conclude for $\LOGSPACE\neq\NL$ and $(\EXT',\kappa)\in\XLnu$ that $(\EXT',\kappa)_\ell\in\LOGSPACE$ holds for every $\ell$. This eventually leads to the desired contradiction proving the theorem.\qed
\end{proof}

%%% Unsinn? EXT(N) in NP!
%\begin{remark}
%  Also the reduction from $\QBF_{2,\exists}$ given by Gottlob in~\cite{gottlob92} to establish the $\SigmaP2$-completeness 
%  of the extension existence problem can be modified to map the given quantified Boolean formula to a default theory 
%  with $W=\emptyset$ and default rules composed of constants and literals only.
%\end{remark}

\section{Autoepistemic Logic}
\label{sect:AL}

Let $B$ be a finite set of Boolean functions. 
To any set $\Sigma$ of autoepistemic $B$-formulae, we associate a 
$\tau_{B,ae} \eqdef  \tau_B \cup \{L^1,\repr^1\}$-structure $\struc{A}_{\Sigma}$ 
such that the universe of $\struc{A}_\Sigma$ is the union of the set of subformulae of 
$\Sigma \cup \{ \neg L\varphi \mid L\varphi \in \SF{\Sigma}\}$, 
the relations from $\tau_B$ are interpreted as in Section~\ref{sect:mso-enc}, and
$L(x)$ holds iff the subformulae represented by $x$ is prefixed by an $L$, and $\repr(x)$ holds iff $x$ represents a formula in $\Sigma$.

\begin{lemma} \label{lem:mso-exp}
  Let $B$ be a finite set of Boolean functions and
  let $\Sigma$ be a set of autoepistemic $B$-formulae.
  There exists an MSO-formula $\theta$ such that $\Sigma$ possesses a stable expansion 
  iff $\struc{A}_{\Sigma} \models \theta$.
\end{lemma}
\begin{proof}
  For a set of formulae $G$ and a formula $\varphi$,
  similar to $\theta_{W\cup C\models\alpha}(C,\alpha)$ in the proof of Lemma~\ref{lem:mso-ext}, define
  be the MSO-formula  
	\begin{align*}
	 \theta_{\Sigma \cup \Lambda \models \varphi}(\Lambda,\varphi)\eqdef \forall M\Big(\theta_{\textit{assign}}(M)\to 
		\forall x\Big(\big((\repr(x)\lor \Lambda(x))\to M(x)\big)\to M(\varphi)\!\Big)\!\Big)
	\end{align*}
  to test for $\Sigma \cup \Lambda \models \varphi$.
  Now define the MSO-formula $\theta_{\full}$ as 
  \begin{align*}
    \theta_{\full}(\Lambda) \eqdef  \, 
      & 
      \forall x\Big(L(x) \imp \big(\Lambda(x) \xor \exists y (\conn_\neg(x,y) \land \Lambda(y))\big)\Big) \land \\
      &
      \forall x\Big(L(x) \imp \big(\Lambda(x) \leftrightarrow \theta_{\Sigma \cup \Lambda \models \varphi}(\Lambda,x) \big)\Big)
  \end{align*}
  Then $\theta \eqdef  \theta_{\textit{struc}} \land \exists \Lambda(\theta_{\full}(\Lambda))$ is true under $\struc{A}_{\Sigma}$ iff $\Sigma$ has a $\Sigma$-full set $\Lambda$, which is the case iff $\Sigma$ has a stable expansion.\qed
\end{proof}

As above we obtain from Lemma~\ref{lem:mso-exp} that, given the tree width of $\struc{A}_{\Sigma}$ as a parameter, 
the expansion existence problem for autoepistemic logic is fixed-parameter tractable, and in fact in $\PLS$.

\begin{theorem} \label{thm:ael-FPT-for-bounded-tw}
  Let $B$ be a finite set of Boolean functions, 
  let $k \in \N$ be fixed, and let $\Sigma$ be a set of autoepistemic $B$-formulae such that $\struc{A}_{\Sigma}$ has tree width bounded by $k$. 
  Then the expansion problem is solvable in time $O(f(k)\cdot|\Sigma|)$ and space $O(\log|\Sigma|)$.
\end{theorem}

On the other hand, analogues of Theorems~\ref{thm:dl-not-in-XP} and \ref{thm:dl-not-in-XL} are easily obtained:
% from 
%either the translation from default logic to full autoepistemic logic given in~\cite{gottlob95} or
%the reduction of the satisfiability problem to the expansion existence problem for sets of autoepistemic $\{\lor\}$-formulae.

\begin{theorem} \label{thm:ael-not-in-XP}
  Let $B$ be a finite set of Boolean functions such that $\lor \in [B \cup \{\false,\true\}]$ and 
  let $\mathbf{\Sigma}$ be the set of sets $\Sigma$ of autoepistemic $B$-formulae such that all $\varphi \in \Sigma$ are disjunctions of propositions or $L$-prefixed propositions.
  Further let $\kappa$ be a parameterization function
  for which 
  there exists a $c \in \N$ 
  such that $\kappa(\Sigma) < c$ for all $\Sigma \in \mathbf{\Sigma}$.
  If $\P\neq\NP$, then the expansion existence problem for sets of autoepistemic $B$-formulae, parameterized by $\kappa$, is not contained in $\XPnu$.
\end{theorem}

\begin{proof}
  Observe that there exists a reduction $f$ from $\threeSAT$ to autoepistemic logic restricted to $B$-formulae shown in Lemma 4.5 of \cite{creignou-meier-thomas-vollmer:10}.
  This implies our claim, as membership in $\XPnu$ implies a polynomial-time algorithm for any fixed $\kappa$.\qed
\end{proof}

\begin{theorem} \label{thm:ael-not-in-XL}
  Let $B$ be a finite set of Boolean functions such that $\xor,\true \in [B]$
  Further let $\kappa$ be a parameterization function
  for which 
  there exists a $c \in \N$ 
  such that $\kappa(\Sigma) < c$ for all $\Sigma$.
  If $\LOGSPACE \neq \ParityL$, then the expansion existence problem for sets of autoepistemic $B$-formulae, parameterized by $\kappa$, is not contained in $\XLnu$.
\end{theorem}

\begin{proof}
  Observe that there exists a reduction $f$ from the implication problem restricted to $B$-formulae shown in Lemma 4.8 of \cite{implicationIPL09}.
  This implies our claim, as membership in $\XLnu$ implies a logspace algorithm for any fixed $\kappa$.\qed
\end{proof}

%% Michael hier fragen woraus das folgt bzw. was gemeint ist.
%Thus, any parameterization that yields fixed-parameter tractable algorithms for the expansion existence problem has to bound the number of stable expansions. 

We remark that similar lower bounds as given for default logic in the previous section and for autoepistemic logic here hold for the implication problem as well%
\ifreport%
, see Appendix~\ref{app:imp}%
\fi%
.

\section{Pseudo-Cliques} \label{sect:pseudo-cliques}

Looking at Theorems~\ref{thm:dl-not-in-XP} and \ref{thm:dl-not-in-XL} one might hope that the syntactic restriction imposed there, namely allowing only defaults that involve literals or propositions, is so severe that it will bound the tree width of every such input structure. Combining this with Theorem~\ref{thm:dl-FPT-for-bounded-tw} would then yield $\P=\NP$ (or $\LOGSPACE=\NL$, resp.).
Stated the other way round, if $\P\neq\NP$ then
%As a consequence of \Cref{thm:dl-FPT-for-bounded-tw,thm:dl-not-in-XP} 
the tree width of $\struc A_{(W,D)}$ is a non-trivial parameterization, \ie, a parameterization $\kappa$ for which there exists no $c\in\N$ such that $\kappa((\emptyset,D))<c$ holds for all $D$ consisting of defaults rules involving only literals. 
%Otherwise the problem $(\EXT',\kappa)$ would not be in $\XPnu$ and would contradict \Cref{thm:dl-FPT-for-bounded-tw}. 

In the following we directly prove the non-triviality of the parameterization by tree width (\ie, without any complexity hypothesizes). As a tool we utilize the subsequent definition of \textit{pseudo-cliques}.

\begin{definition}\label{def:pseudo-clique}
 Let $G=(V,E)$ be an undirected graph. A \emph{pseudo-clique} is a set of vertices $V'\subseteq V$ that can be partitioned into the set of \emph{main-nodes} $V_{\textit{main}}$ and sets of \emph{edge-nodes} $V_{u,v}$ for each $u\neq v\in V_{\textit{main}}$ such that the following holds: for $v_1,\dots,v_m\in V_{u,v}$ the nodes in $V_{u,v}$ form a simple path from $u$ to $v$, \ie, it holds that $(u,v_1),(v_1,v_2),\dots,(v_{m-1},v_m),(v_m,v)\in E$ and no other edges are present.
 
 The \emph{size} of a pseudo-clique is $|V_{\textit{main}}|$, \ie, the number of main-nodes.
 The \emph{cardinality} of a pseudo-clique is $\max_{u\neq v\in V_{\textit{main}}}{|V_{u,v}|}$, \ie, the length of the longest simple path between edge-nodes.
 A pseudo-clique is said to have \emph{exact cardinality} $k$
if $\forall u,v\in V_{main}$: $|V_{u,v}| = k$.
% \ie, between each pair of main-nodes the simple path has exactly length $k$.
% Note that for cardinality 0 we obtain just the classical clique definition.
\end{definition}

The first five pseudo-cliques of exact cardinality $1$, and one of cardinality 3 are visualized in Figure~\ref{fig:pseudo-cliques}. The thick vertices correspond to the main-nodes whereas the small dots correspond to the edge-nodes.

\begin{figure}
	\centering
\includegraphics{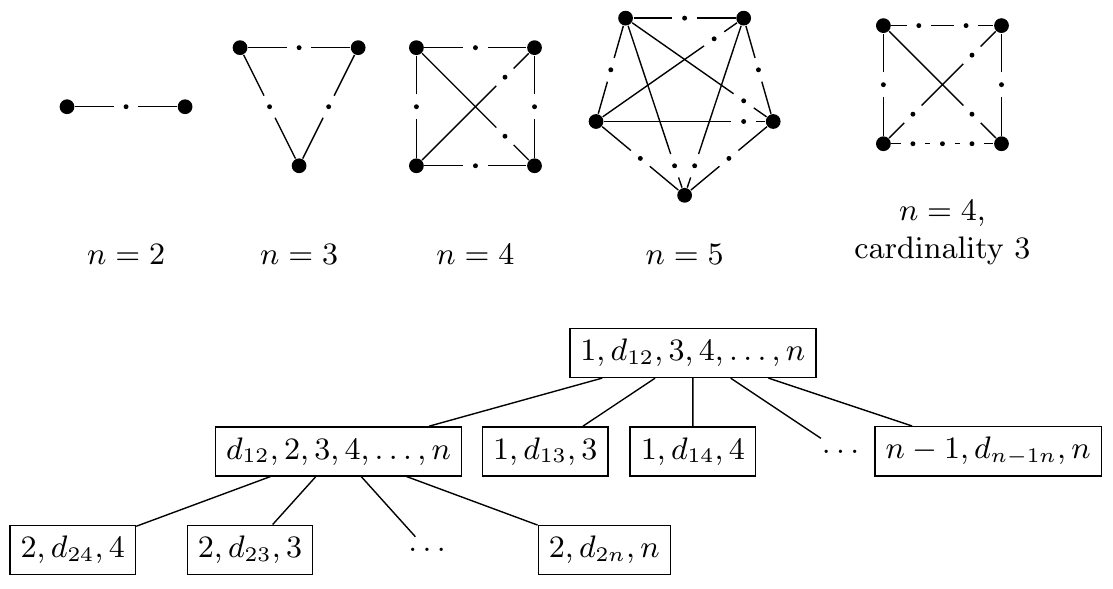}
	\caption{Pseudo-cliques of exact card. $1$ and size $\in \{2,\dots,5\}$, one of card. $3$ and size $4$, and a tree decomposition of a pseudo-clique of exact card. $1$ and size $n$.}\label{fig:pseudo-cliques}
\end{figure}

The important fact for us is the observation that pseudo-cliques of size $n$ have the same tree width as the clique of size $n$.

\begin{theorem}\label{thm:tw-pseudo-cliques-cardinality}
	Let $G=(V,E)$ be a pseudo-clique of size $n\geq 3$ and cardinality $k\geq 0$. Then the tree width of $G_n$ is $n-1$.
\end{theorem}
\ifreport
\begin{proof}
 Is proven in Appendix~\ref{app:pseudo}.
\end{proof}
\fi

Whenever one wants to show that a parametrization by tree width is non-trivial, the most obvious method is to show that the family of graphs has (sub-) cliques of arbitrary size.
Now Theorem~\ref{thm:tw-pseudo-cliques-cardinality} provides an alternative when this method is prohibited: it suffices to construct pseudo-cliques. Corollary~\ref{cor:tw-nontriv-DL} (\ref{num:no-cliques}.) shows that, for families used for the lower bounds in the previous sections, it is not possible to use cliques in order to prove  unbounded tree width  and therefore additionally motivates the definition and purpose of pseudo-cliques.

\begin{corollary}\label{cor:tw-nontriv-DL}Let $B$ be a finite set of Boolean functions such that $\false \in [B]$ and let $(\emptyset,D)$ be a $B$-default theory in the sense of Theorem~\ref{thm:dl-not-in-XP}, i.e., each default in $D$ is composed of literals only. Then there exists an MSO formula $\theta$ fulfilling the property $(W,D)\in\EXT(B)$ iff $\struc A_{(W,D)}\models\theta$, and
	\begin{enumerate}
		\item $\struc A_{(W,D)}$ is neither $\ell$-connected nor contains a clique of size $\ell$ for any $\ell\ge3$.\label{num:no-cliques}
		\item There exists a family of default theories $(\emptyset,D)_k$ such that the tree width of $\struc A_{(\emptyset,D)_k}$ is not constant.
	\end{enumerate}
\end{corollary}
\begin{proof}
For (1.) we construct the MSO formula $\theta$ according to Lemma~\ref{lem:mso-ext}. At first observe that the universe of $\struc A_{(W,D)}$ comprises only literals and defaults. Further, there are no edges between literals, and no edges between defaults. Every default can be connected to at most three different literals. Obviously the graph does not contain a clique of size $\ell\ge3$. Furthermore, the graph is not $\ell$-connected for any $\ell\ge3$ by the following observation. Let $x_d$ be some individual representing the default $d=\frac{\alpha:\beta}{\gamma}$. Then there are individuals $x_\alpha,x_\beta,x_\gamma$ to represent the respective parts of $d$ which are all connected to $x_d$. If now $x_\alpha,x_\beta$ and $x_\gamma$ are removed from the graph, then there is no other individual to which $x_d$ is connected yielding a contradiction to the connectivity. 

Turning to (2.) observe that (1.) prohibits using $\ell$-cliques or $\ell$-connectivity for any $\ell\ge3$ to measure the tree width of $\struc A_{(W,D)_k}$. Now define a default theory $(\emptyset,D)$ complying with Theorem~\ref{thm:ael-not-in-XP}, where 
$
 D \eqdef \left.\left\{d_{ij}=\frac{x_i:y_j}{\false} \;\right|\; 1\leq i\leq j\leq n\right\},
$
and $x_i,y_j$ are variables for $1\leq i\leq j\leq n$. Application of Theorem~\ref{thm:tw-pseudo-cliques-cardinality} concludes the proof.
\qed
\end{proof}

% Thus the intuition that tree width as parameterization cannot be used in the sense of \Cref{thm:dl-not-in-XP} in order to state complexity collapses is correct. 

\noindent An analogous result holds for autoepistemic logic.

\begin{corollary}\label{cor:tw-nontriv-AEL}Let $B$ be a finite set of Boolean functions.
There exists a family of autoepistemic $B$-formulae $\Sigma_k$ and all $\varphi\in\Sigma_k$ are disjunctions of propositions or $L$-prefixed propositions such that there exists an MSO formula $\theta$ fulfilling the property $\Sigma_k\in\EXP(B)$ iff $\struc A_{\Sigma_k}\models\theta$ and the tree width of $\struc A_{\Sigma_k}$ is not constant.
\end{corollary}

\begin{proof}
 Define $\Sigma_k$ as
	 $\Sigma_k \eqdef \{x_i\lor x_j\mid 1\leq i \leq j\leq k\}.$ 
	Then the structure $\struc A_{\Sigma_k}$ consist of cliques of size $k$, in fact.\qed 
\end{proof}

\begin{corollary}\label{cor:tw-nontriv-IMP}Let $B$ be a finite set of Boolean functions such that $\land,\lor\in[B]$. Let $\mathbf{\Gamma}_1$ be the set of sets $\Gamma$ of formulae in monotone $2$-CNF and let $\mathbf{\Gamma}_2$ be the set of sets $\Gamma$ of formulae in DNF.
There exists a family of sets of $B$-formulae $(F,G)_k$ with $F\in\mathbf{\Gamma}_1, G\in\mathbf{\Gamma}_2$ such that there exists an MSO formula $\theta$ fulfilling the property $(F,G)_k\in\IMP(B)$ iff $\struc A_{(F,G)_k}\models\theta$ and the tree width of $\struc A_{(F,G)_k}$ is not constant.
\end{corollary}

\section{Conclusion} \label{sect:conclusion}

In this paper we applied Courcelle's Theorem~\cite{courcelle90} and the logspace version of Elberfeld\,\etal\cite{ebjata10} to the most prominent decision problems in the nonmontonic default logic and autoepistemic logic. Thereby we showed that the extension existence problem for a given default theory $(W,D)$ is solvable in time $O(f(k)\cdot|(W,D)|)$ and space $O(\log|(W,D)|)$ if the tree width of the corresponding MSO structure is bounded by $k$; similarly for the expansion existence problem for a set of autoepistemic formulae, and as well for the implication problem for sets of formulae $F,G$.

We mention that furthermore one can achieve similar results for the credulous (resp. brave) and skeptical (resp. cautious) reasoning problems of the nonmontone logics from above by slight extensions of the constructed MSO-formulae.

Furthermore we introduced with \textit{pseudo-cliques} a weaker notion of cliques: basically we have a clique where each edge is divided into two edges by a fresh node  (or even a longer path). There we showed that the tree width of a graph is bounded from below by the size of its largest sub-pseudo-clique. If we investigate default theories $(W,D)$ which contain an empty knowledge base $W$ and only defaults which are composed of propositions or the constant $\false$ only, then for constant parameterizations we show collapses of $\P$ and $\NP$ (resp. $\LOGSPACE$ and $\NL$) if the corresponding parameterized problem is in $\XPnu$ (resp. $\XLnu$).
Thus through the concept of pseudo-cliques we construct a family of default theories whose tree width of its MSO-structures is unbounded. Therefore this kind of parameterization cannot be used to prove such complexity class collapses. Analogue claims can be made for the expansion existence problem in autoepistemic logic and the implication problem for sets of formulae.

For further research it would be very interesting to find a parameterization that is non-trivial in the sense of Theorem~\ref{thm:dl-not-in-XP} but uses many different values. Also insights on new types of parameterizations, in particular in the context of the new space parameterized complexity classes, would be very engaging.

\paragraph{Acknowledgement.}

For helpful hints and discussions we are grateful to Nadia Creignou (Marseille) 
and Thomas Schneider (Bremen).

\bibliographystyle{splncs}
\bibliography{nml-fpt}

\ifreport
\newpage
\appendix

\section{Lower Bounds for the Implication Problem}
\label{app:imp}

Analogues of Theorems~\ref{thm:dl-not-in-XP} and \ref{thm:dl-not-in-XL} for autoepistemic logic were given in Section~\ref{sect:AL}. Here, we want to point out that also for the implication problems, similar results hold.

%Further we are able to formulate two results emerging around the classes $\XLnu$ and $\XPnu$.

\begin{theorem} \label{thm:imp-not-in-XL}
  Let $B$ be a finite set of Boolean functions such that $x\xor y\xor z \in [B]$ and
  let $\mathbf{\Gamma}$ be the set of sets $\Gamma$ of formulae such that each formula  
  $\varphi \in \Gamma$ is composed of functions $f(x,y,z)=x\xor y\xor z$ only.
  Further let $\kappa$ be a parameterization function
  for which 
  there exists a $c \in \N$ 
  such that $\kappa\big(F,G\big) < c$ for all $F,G \in \mathbf{\Gamma}$.
  If $\LOGSPACE\neq\ParityL$, then the implication problem for sets of $B$-formulae, parameterized by $\kappa$, is not contained in $\XLnu$.
\end{theorem}

\begin{proof}
 From Lemma 4.4 in \cite{implicationIPL09} it follows that the implication problem $\IMP'$ for these sets of formulae is $\ParityL$-hard. Suppose $\LOGSPACE\neq\ParityL$ and let $\kappa$ be a parameterization such that $\kappa(F,G)<c$ for every $F,G\in\mathbf{\Gamma}$. Now following the argumentation of Theorem~\ref{thm:dl-not-in-XP} yields a contradiction to $\LOGSPACE\neq\ParityL$.
% For contradiction assume $(\IMP',\kappa)\in\XLnu$. Hence, by definition of $\XLnu$, it holds that $(\IMP',\kappa)_\ell\in\L$ for all $\ell\in\N$. As $\ell<c$ holds due to the choice of $\kappa$ we can compose a deterministic algorithm running in logarithmic space solving $\IMP'$ due to $\L$ being closed under finite union. This is a contradiction to $\L\neq\ParityL$ and concludes the proof.\qed
\end{proof}

\begin{theorem} \label{thm:imp-not-in-XP}
  Let $B$ be a finite set of Boolean functions such that $\land,\lor \in [B]$. 
  Let $\mathbf{\Gamma}_1$ be the set of sets $\Gamma$ of formulae in monotone $2$-CNF and let $\mathbf{\Gamma}_2$ be the set of sets $\Gamma$ of formulae in DNF.
  Further let $\kappa$ be a parameterization function
  for which 
  there exists a $c \in \N$ 
  such that $\kappa\big(F,G\big) < c$ for all $F \in \mathbf{\Gamma}_1,G\in\mathbf\Gamma_2$.
  If $\P\neq\NP$, then the implication problem for sets of $B$-formulae, parameterized by $\kappa$, is not contained in $\XPnu$.
\end{theorem}

\begin{proof}
 From Lemma 4.2 in \cite{implicationIPL09} it follows that the implication problem $\IMP'$ for these sets of formulae is $\coNP$-hard. Following an analogue argumentation as in the proof of Theorem~\ref{thm:imp-not-in-XL} implies this theorem.\qed
\end{proof}

\section{Tree Width of Pseudo-Cliques}
\label{app:pseudo}

%We say that a pseudo-clique has \emph{exact cardinality} $k$, if $\forall u,v\in V'$: $|V_{main}| = k$, \ie, between each pair of main nodes the simple path contains at least one edge node.

\begin{theorem}\label{thm:pseudo-clique-tw}
 Let $G_n=(V,E)$ be a pseudo-clique of size $n\geq 3$ and cardinality $k\geq 0$. 
 Then the tree width of $G_n$ is $n-1$.
\end{theorem}
\begin{proof}
 Let $G_n=(V,E)$ be the following undirected graph, where
\begin{itemize}
	\item $V$ contains $n$ \emph{main-nodes}, labeled $i$ for $1 \leq i \leq n$ and a number of \emph{edge-nodes}
	  labeled $d^r_{i,j}$ for $1 \leq i < j \leq n$ and $1 \leq r \leq k_{i,j}$, $k_{i,j} \leq k$, and
	  % $\{d_{1,2}, \dots ,d_{1,n}, d_{2,3}, \dots, d_{2,n}, \dots, d_{n-1,n}\}$.
	\item $E = \{\{i,d^1_{i,j}\}, \{d^1_{i,j},d^2_{i,j}\}, \dots , \{d^{k_{i,j}-1}_{i,j},d^{k_{i,j}}_{i,j}\}, \{d^{k_{i,j}}_{i,j},j\}\;|\; 1 \leq i < j \leq n\}$.
\end{itemize}

\begin{claim}
Let $n \geq 3$. Any tree decomposition $(T,X)$ of $G_n = (V,E)$ can be transformed into a tree decomposition $(T',X')$ of $G_n$ such that
\begin{enumerate}
	\item $\width(T',X') \leq \width(T,X)$,
	\item in $(T',X')$ every edge-node $d^r_{i,j}$ appears exactly once, in particular in a bag of size $\leq 3$.
\end{enumerate}
\end{claim}
\begin{claimproof}
We execute consecutively for every pair of main-nodes $(i,j)$ the following procedure on the (valid) tree decomposition $(T,X)$.
\begin{quote}
	Consider $Y = \{B \in X\; |\; d^r_{i,j} \in B\}$, \ie, all bags containing edge-nodes between $i$ and $j$.
	Observe that all bags in $Y$ are connected in $T$.
% since $(T,X)$ is a valid tree decomposition.
	Pick up one bag $C \in Y$ such that $i \in C$. Such a $C$ does exist, since $(T,X)$ is a valid tree decomposition
	and $\{i,d^1_{i,j}\}\in E$. Now replace in
	every bag of $Y$ every occurrence of any edge-node $d^r_{i,j}$ by $j$. Add to $C$ the hierarchical chain of children
	$\{i,d^1_{i,j},j\}, \{d^1_{i,j},d^2_{i,j},j\}, \dots,  \{d^{k_{i,j}-1}_{i,j}, d^{k_{i,j}}_{i,j},j\}, \{d^{k_{i,j}}_{i,j},j\}$.
\end{quote}
One verifies that after each \emph{round} of this procedure the so obtained structure $(T',X')$ is a valid tree decomposition of $G_n$ satisfying
condition 1. Finally, it is obvious that after the whole procedure $(T',X')$ satisfies even condition 2.
\end{claimproof}
Further, the following claim shows that the tree width of $G_n$ depends on $n$.

\begin{claim}
The graph $G_n$ has tree width $n-1$.
\end{claim}
\begin{claimproof}
It suffices to show that the transformed tree decomposition $(T',X')$ of the preceding claim has width $n-1$.
We observe that all bags in $X'$ of size greater than $3$ are bags containing solely main-nodes.
Thus we have $\width(T',X') \leq n-1$. 

For a lower bound observe that for each pair $(i,j)$ there is at least one bag $B \in X'$ such that $i,j \in B$. Therefore, $(T',X')$ restricted to bags containing only main-nodes (\ie, any $d^r_{i,j}$ is removed from its bag) is a valid tree decomposition of the edge-complete graph with node set $\{1,\dots,n\}$. Thus, $\width(T',X') \geq n-1$.
%Since we have $\tw(T',X') \leq \tw(T,X)$ for \emph{any} tree decomposition $(T,X)$, we conclude that $\tw(G_n) = \tw(T',X') \geq n-1$.
\end{claimproof}
This concludes the proof.\qed
\end{proof}
\fi
\end{document}